\documentclass[]{amsart}
\usepackage{amssymb}
\usepackage[dvips]{epsfig}
\usepackage{graphicx}
\usepackage{color}
\usepackage{amsmath}
\usepackage{amssymb}
\usepackage{amsfonts}
\usepackage{amsthm}

\newcommand{\R}{\mathbb R}

\newcommand{\Z}{\mathbb Z}

\newcommand{\N}{\mathbb{N}}

\newtheorem{thm}{Theorem}[section]
\newtheorem{lem}[thm]{Lemma}
\newtheorem{prop}[thm]{Proposition}

\theoremstyle{remark}
\newtheorem{rem}{\bf Remark}[section]
\theoremstyle{definition}
\newtheorem{defn}[thm]{Definition}

\numberwithin{equation}{section}
\usepackage[colorlinks=true,pdfstartview=FitV,linkcolor=magenta,citecolor=cyan]{hyperref}
\usepackage{bm}

\begin{document}
\title[]{Long-time Anderson Localization for the  Nonlinear quasi-periodic Schr\"odinger \\Equation on $\mathbb Z^d$}
\author{Hongzi Cong}
\address[HC]{School of Mathematical Sciences, Dalian University of Technology, Dalian 116024, China}
\email{conghongzi@dlut.edu.cn}

\author[]{Yunfeng Shi}
\address[YS]{School of Mathematics,
Sichuan University,
Chengdu 610064,
China}
\email{yunfengshi@scu.edu.cn}
\author{W.-M. Wang}
\address[WW] {CNRS and D\'epartment De Math\'ematique,
	Cergy Paris Universit\'e,
	Cergy-Pontoise Cedex 95302,
	France}
\email{wei-min.wang@math.cnrs.fr}

\date{\today}



\begin{abstract} Using a Birkhoff normal form transform to impede mode transfer in a finite ``barrier", we prove localization of {\it arbitrary} $\ell^2$ data
for polynomially long time for the nonlinear quasi-periodic Schr\"odinger equation on $\mathbb Z^d$. 

\end{abstract}

\maketitle

\section{Introduction}

We consider the nonlinear Schr\"odinger equation on $\mathbb Z^d$: 
\begin{equation}\label{010701}
	{\rm{i}} \dot{q}_{\bm j}=V_{\bm j}(\bm \theta,\bm \alpha)q_{\bm j}+\epsilon_1(\Delta q)_{\bm j}+\epsilon_2|q_{\bm j}|^2q_{\bm j},\, {\bm j}\in\mathbb Z^d, 
\end{equation}
where $\epsilon_1$ and $\epsilon_2$ are parameters in $[0, 1]$, and 
$$(\Delta q)_{\bm j}=\sum_{\bm j', | \bm j'-\bm j|_1 =1} q_{\bm j'},$$
is the usual discrete Laplacian, $|\bm j|_1=\sum_{1\leq i\leq d}|j_i|$ and 
the potential $V_{\bm j}(\bm \theta,\bm \alpha)$ is a trigonometric polynomial given by
\begin{align}\label{080401}
	V_{\bm j}(\bm\theta,\bm\alpha)=\sum_{\bm \ell \in \Gamma_L}v_{\bm\ell}\cos2\pi\left(\bm \ell  \cdot\left(\bm\theta+\bm{j}\bm{\alpha}\right)\right)\ {\rm with}\  \bm \theta, \bm\alpha \in [0,1]^d, 
\end{align}
where 
\begin{align*}
\bm x\bm y&:=(x_1y_1,\cdots, x_d y_d)\in \R^d,\\
\bm x\cdot\bm y&:=x_1y_1+\cdots+x_dy_d\in\R,
\end{align*}
for $\bm x,\bm{y}\in\R^d$, and $\Gamma_L\subset \mathbb{Z}^d,\ L\in\mathbb{N}$ satisfies  the following properties: 
\begin{equation*}\label{Gammk}
	\begin{aligned}
		&\text {(a) for each }\bm \ell=(\ell_k)_{1\leq k\leq d}\in\Gamma_L, \ell_k\neq 0\ {\rm for}\  \forall \ k=1, 2, ..., d; \\
		&\text {(b) for any two }\ \bm\ell,\bm\ell'\in\Gamma_L, \, \bm\ell+\bm\ell'\neq 0.
	\end{aligned}
\end{equation*}
Note that finite multi-variable cosine series naturally satisfy conditions (a) and (b).

When $\epsilon_2=0$, it is known from the work of Bourgain \cite {Bou07} that for analytic $V$, for any fixed $\bm \theta$, there 
is a large set in $\bm \alpha$, so that the linear Schr\"odinger operator $H$:
$$H=V(\bm \theta,\bm \alpha)+\epsilon_1\Delta $$
has Anderson localization if  $\epsilon_1$ is sufficiently small. (See also \cite{BGS02} for $\mathbb Z^2$; and \cite{JLS20} for  general long-range operators on $\mathbb Z^d$.)
Consequently using eigenfunction expansion, all $\ell^2$ solutions which are initially localized about 
the origin, remain localized for all time.

The main purpose of this paper is to prove an analogous, long time result, for the nonlinear equation \eqref{010701}, when $\epsilon_2\neq 0$. 
Our main result is the following: 
\begin{thm}\label{main}
	Given any $\delta,\gamma>0,$ $M\gg1$ and for all initial datum $q(0)\in\ell^2(\mathbb{Z}^d)$, let $j_0\in\mathbb{N}$ be such that
	\begin{equation*}
		\sum_{|\bm j|>j_0}|q_j(0)|^2<\delta,
	\end{equation*}
where $|\bm j|=\sqrt{\sum_{1\leq i\leq d}|j_i|^2}$.
	Then there exists a constant $\varepsilon(\gamma,L,M,j_0)>0$ so that the following holds: for $0<\epsilon:=\epsilon_1+\epsilon_2<\varepsilon(\gamma,L,M,j_0)$ and
	for all
	$$|t|\leq \delta\cdot \epsilon^{-M}$$ one has
	\begin{equation*}
		\sum_{|\bm j|>j_0+M^2}|q_{j}(t)|^2<2\delta,
	\end{equation*}
on a set in $(\bm\theta,\bm\alpha)$ of measure at least
	\begin{align*}
		1-\gamma.
	\end{align*}
\end{thm}

\subsection{Ideas of the proof} We use a Birkhoff normal form transform to prove the Theorem. This transform is related to that in \cite{BW07, WZ09, CSZ21} in the random setting.
(For physical motivations of the problem, see e.g., \cite{FKS12}.)
Here we extend it to the quasi-periodic setting, and to arbitrary dimensions. Similar to  \cite{BW07, WZ09, CSZ21}, this normal form transform is effectuated in a finite neighborhood, 
seen as a ``barrier" to impede transfer to higher modes, higher indexed lattice sites.  

Small divisor conditions are imposed, as usual, in the barrier for the normal form transform. The measure estimate of these small divisor conditions is more difficult 
than that in the random case. We use the generalized Wronskian approach initiated in \cite{SW23-1, SW23-2}.  This is the main novelty here.

In \cite{SW23-1}, the Wronskian method was first used for the cosine potential
in one dimensional phase space, (i.e, the $\bm\theta$ in \eqref{010701} is one dimensional).
 As it turns out, however, this method is general, and applicable to arbitrary trigonometric polynomials in arbitrary dimensional
phase space \cite{SW23-2}. Here we follow \cite{SW23-2} to make small divisor estimates and to prove that arbitrary $\ell^2$ data remain localized, in the sense of the Theorem, for polynomially long time.

\subsection{Global in time solutions to quasi-periodic systems} 
It is natural to enquire whether nonlinear quasi-periodic systems such as equation \eqref{010701} have {\it  global} solutions. In \cite{SW23-1}, quasi-periodic in time 
solutions were constructed for the nonlinear quasi-periodic wave equations (NLQPW)  in $\mathbb Z^d$, for the cosine potential (in one dimensional phase space). 
 Related results for the nonlinear quasi-periodic Schr\"odinger equation on $\mathbb Z$ were known before, see e.g., \cite{GYZ14, Yua02}. But the existence of quasi-periodic in time
 solutions to NLQPW seemed to have remained open until \cite{SW23-1}. Finally, we would like to mention that the recent result in \cite{SW23-2} should pave the way toward 
 the existence of quasi-periodic in time solutions to \eqref{010701}, and also its NLQPW counterpart, in $\mathbb Z^d$, for arbitrary trigonometric polynomial potentials, thus further generalizing Bourgain's work \cite{Bou07} to a nonlinear setting. 
 
 \subsection{Some details about the small divisors}
Before closing the section, let us give a bit more precision on the small divisors, which the reader may skip as a first reading. As mentioned earlier, the Birkhoff normal form transform is in a finite neighborhood, more precisely, in  
an annulus neighborhood of the sphere $A(j_0):=\left\{\bm j\in\mathbb{Z}^d:\left|\bm j\right|=j_0\right\}$, 
which then acts as a ``barrier" to impede mode transfer. 
We choose ``good" $(\bm\theta,\bm\alpha)$ so that $V=(V_{\bm j}(\bm \theta,\bm \alpha))_{\bm j\in\mathbb{Z}^d}$ satisfies suitable \textbf{non-resonant conditions}. 

More precisely, in the Birkhoff normal form transform, the values of the potential $V$ at different sites $\bm j$, $V_{\bm j}$ act as components of a frequency vector.  Given $\gamma,L,M,j_0>0$, we say that a frequency vector $\omega=(\omega_{\bm j})_{\bm{j}\in\mathbb{Z}^d}$ is $(\gamma,L,M,j_0)$-nonresonant if for any $0\neq \bm k=(k_{\bm j})_{j\in\Z^d}\in\mathbb{Z}^{\mathbb{Z}^d}$,
	\begin{align}\label{122106}
		\left|\sum_{\bm j\in\mathbb{Z}^d}k_{\bm j}\omega_{\bm j}\right|\geq\left(\frac{\gamma}{\left(2LM^2j_0\right)^{2(d+1)}}\right)^{10L^4M^4}.
	\end{align}
Since the normal form transform is in the barrier only, it suffices to realize the above non-resonant condition near $A(j_0)$.
We remark that this is an essential point in the proof.

\smallskip


The structure of the paper is as follows.  Some important facts on Hamiltonian dynamics, such as the Poisson bracket, symplectic transformation and non-resonant conditions  
 are presented in \S2. The Birkhoff normal form type theorem and the main theorem are proved in \S3.  The measure estimates for the small divisors are in \S4.  Subsequently \S5 concludes
 the proof of some technical lemmas.

\section{Structure of the transformed Hamiltonian}
We recast (\ref{010701}) as a Hamiltonian equation
\begin{equation*}
\rm{i} \dot{q}_{\bm j}=2\frac{\partial H}{\partial \bar{q}_{\bm j}},
\end{equation*}
where
\begin{align}\label{010703}
H(q,\bar q)=\frac{1}2\left(\sum_{\bm j\in\mathbb{Z}^d}V_{\bm{j}}(\bm\theta,\bm\alpha)\left|q_{\bm j}\right|^2+\epsilon_1\sum_{\bm j,\bm j'\in\mathbb{Z}^d\atop \left|\bm j-\bm j'\right|=1}q_{\bm j}\bar q_{\bm j'}+\frac12\epsilon_2\sum_{\bm j\in\mathbb{Z}^d}\left|q_{\bm j}\right|^4\right),
\end{align}
where $V_{\bm{j}}(\bm\theta,\bm\alpha)$ is defined by (\ref{080401}).
In order to prove the main result, we need to control the time derivative of the truncated sum of higher modes
\begin{align}\label{022101}
\frac{d}{dt}\sum_{|\bm j|>j_0}\left|q_{\bm j}(t)\right|^2.
\end{align}
In what follows, we will deal extensively with monomials in $q_{\bm j}$.  Rewrite any monomial in the form
\begin{align}\label{mono}
\prod_{\bm j\in\mathbb{Z}^d}q_{\bm j}^{n_{\bm j}}\bar q_{\bm j}^{n_{\bm j}'}.
\end{align}
Let
\begin{align*}
\bm n=(\bm n_{\bm j},\bm n_{\bm j}')_{\bm j\in\mathbb{Z}^d}\in \mathbb{N}^{\mathbb{Z}^d}\times\mathbb{N}^{\mathbb{Z}^d}.
\end{align*}
We  define
\begin{align*}
{\rm supp}\  \bm n&=\{\bm j\in\mathbb{Z}^d: \ \bm n_{\bm j}\neq 0\ \mbox{or}\ \bm n_{\bm j}'\neq 0\},\\
\Delta(\bm  n)&=\sup_{\bm j,\bm j'\in{\rm supp}\  \bm n }|\bm j-\bm j'|,\\
|\bm n|_1&=\sum_{\bm j\in\mathbb{Z}^d}(\bm n_{\bm j}+\bm n_{\bm j}').
\end{align*}
If $\bm n_{\bm j}=\bm n_{\bm j}'$ for all $\bm j\in\mbox{supp}\ \bm n$, then the monomial \eqref{mono} is called resonant. Otherwise it is non-resonant. Note that non-resonant monomials contribute to the truncated sum in (\ref{022101}); while resonant ones do not. We define the (resonant) set as
\begin{align}\label{030301}
\mathcal{N}=\left\{\bm n\in\mathbb{N}^{\mathbb{Z}^d}\times\mathbb{N}^{\mathbb{Z}^d}:\ \bm n_{\bm j}=\bm n_{\bm j}'\ {\rm for}\ \forall\  \bm j\right\}.
\end{align}
Given a large $j_0\in\mathbb{N}$ and $N\in\mathbb{N}$ satisfying $N\ll  j_0$, let
\begin{align*}
A(j_0,N):=\left\{\bm j\in\mathbb{Z}^d:\left|\left|\bm j\right|-j_0\right|\leq N\right\}.
\end{align*}

\begin{defn}\label{122403}
Given a Hamiltonian
\begin{align*}
W(q,\bar{q})=\sum_{\bm n\in\mathbb{N}^{\mathbb{Z}^d}\times\mathbb{N}^{\mathbb{Z}^d}}W(\bm n)\prod_{{\rm supp}\ \bm n}q_{\bm j}^{\bm n_{\bm j}}\bar {q}_{\bm j}^{\bm n_{{\bm j}}'},
\end{align*}
for $j_0,N\in\mathbb{N}$ and $r>2$, we define
\begin{align}\label{122402}
\left\| W\right\|_{j_0,N,r}=\sum_{\bm n\in\mathbb{N}^{\mathbb{Z}^d}\times\mathbb{N}^{\mathbb{Z}^d}\atop
{{\rm supp}}\ \bm n\cap A(j_0,N)\neq \emptyset}{\left|W(\bm n)\right|\cdot|\bm n|_1\cdot r^{\Delta(\bm n)+|\bm n|_1-1}}.
\end{align}
\end{defn}

\begin{rem}
	It is easy to see that
	\begin{equation*}
	\left\| W\right\|_{j_0,N_1,r_1}\leq \left\| W\right\|_{j_0,N_2,r_2},
	\end{equation*}
if 
\begin{equation*}
	N_1\leq N_2\quad \mbox{and}\quad r_1\leq r_2.
\end{equation*}
\end{rem}

\begin{defn}
The Poisson bracket of $W$ and $U$ is defined as
$$
\{W,U\}:={\rm{i}}\sum_{\bm j\in{\mathbb{Z}^d}}\left(\frac{\partial W}{\partial q_{\bm j}}\cdot\frac{\partial U}{\partial \bar q_{\bm j}}-\frac{\partial W}{\partial \bar q_{\bm j}}\cdot\frac{\partial U}{\partial q_{\bm j}}\right).
$$
\end{defn}

We have the following key estimate.

\begin{prop}[\textbf{Poisson Bracket}]\label{090603}
For $j_0,N\in\mathbb{N}$, let 
\begin{align*}
	W(q,\bar{q})=\sum_{\bm n\in\mathbb{N}^{\mathbb{Z}^d}\times\mathbb{N}^{\mathbb{Z}^d}\atop {\rm supp}\ \bm n\subset A(j_0,N)}W(\bm n)\prod_{{\rm supp}\ \bm n}q_{\bm j}^{\bm n_{\bm j}}\bar {q}_{\bm j}^{\bm n_{\bm j}'},
\end{align*}
and
\begin{equation*}
	U(q,\bar{q})=\sum_{\bm m\in\mathbb{N}^{\mathbb{Z}^d}\times\mathbb{N}^{\mathbb{Z}^d}}U(\bm m)\prod_{{\rm supp}\ \bm m}q_{\bm j}^{\bm m_{\bm j}}\bar {q}_{\bm j}^{\bm m_{\bm j}'}.	
\end{equation*}
Then for any $0<\sigma<r/2$, we have
\begin{align}\label{122002}
\left\|\left\{ W, U\right\}\right\|_{j_0,N,r-\sigma}\leq \frac{1}\sigma\left\|W\right\|_{j_0,N,r}\cdot\left\|U\right\|_{j_0,N,r}.
\end{align}
\end{prop}

\begin{proof}
First of all, we write
\begin{align*}
\{W,U\}=\sum_{\bm l\in\mathbb{N}^{\mathbb{Z}^d}\times\mathbb{N}^{\mathbb{Z}^d}}\{W,U\}
(\bm l)\prod_{{\rm supp}\ \bm l}q_{\bm j}^{\bm l_{\bm j}}\bar{q}_{\bm j}^{\bm l_{\bm j}'}.
\end{align*}
Then one has
\begin{align}\label{011401}
\{W,U\}
(\bm l)={\rm{i}}\sum_{\bm k\in\mathbb{Z}^d}\left(\sum_{\bm n,\bm m\in\mathbb{N}^{\mathbb{Z}^d}\times\mathbb{N}^{\mathbb{Z}^d}}^*W(\bm n)U(\bm m)\left(\bm n_{\bm k}\bm m'_{\bm k}-\bm n_{\bm k}'\bm m_{\bm k}\right) \right)
\end{align}
and the sum $\sum\limits_{\bm n,\bm m\in\mathbb{N}^{\mathbb{Z}^d}\times\mathbb{N}^{\mathbb{Z}^d}}^*$  is taken as
\begin{align*}
&\bm l_{\bm j}=\bm n_{\bm j}+\bm m_{\bm j}-1, \quad \bm l'_{\bm j}=\bm n'_{\bm j}+\bm m'_{\bm j}-1\ {\rm for}\ \bm j=\bm k,\\
&\bm l_{\bm j}=\bm n_{\bm j}+\bm m_{\bm j},\quad  \bm l'_{\bm j}=\bm n'_{\bm j}+\bm m'_{\bm j }\  {\rm for}\ \bm j\neq \bm k.
\end{align*}

Secondly, let
\begin{align*}
\widetilde U=\sum_{\bm m\in\mathbb{N}^{\mathbb{Z}^d}\times\mathbb{N}^{\mathbb{Z}^d}\atop
{\rm supp}\ \bm m\cap A(j_0,N)=\emptyset}U(\bm m)\prod_{{\rm supp}\ \bm m}q_{\bm j}^{\bm m_{\bm j}}\bar {q}_{\bm j}^{\bm m'_{\bm{j}}}, 
\end{align*}
then one has $\left\{W,\widetilde U\right\}=0.$
Hence, we can always assume that
\begin{align}\label{022401}
U=\sum_{\bm m\in\mathbb{N}^{\mathbb{Z}^d}\times\mathbb{N}^{\mathbb{Z}^d}\atop
	{\rm supp}\ \bm m\cap A(j_0,N)\neq\emptyset}U(\bm m)\prod_{{\rm supp}\ \bm m}q_{\bm j}^{\bm m_{\bm j}}\bar {q}_{\bm j}^{\bm m'_{\bm{j}}}.
\end{align}

Without loss of generality, we assume that $W$ and $U$ are homogeneous polynomials with degrees $n^*$ and $m^*$ respectively, i.e.,
\begin{align*}
W(q,\bar{q})=\sum_{\bm n\in\mathbb{N}^{\mathbb{Z}^d}\times\mathbb{N}^{\mathbb{Z}^d},\left|\bm n\right|_1=n^*\atop {\rm supp}\ \bm n\subset A(j_0,N)}W(\bm n)\prod_{{\rm supp}\ \bm n}q_{\bm j}^{\bm n_{\bm j}}\bar {q}_{\bm j}^{\bm n_{\bm j}'},
\end{align*}
and
\begin{align*}
U(q,\bar{q})=\sum_{\bm m\in\mathbb{N}^{\mathbb{Z}^d}\times\mathbb{N}^{\mathbb{Z}^d},\left|\bm m\right |_1=m^*\atop
{\rm supp}\ \bm m\cap A(j_0,N)\neq \emptyset}U(\bm m)\prod_{{\rm supp}\ \bm m}q_{\bm j}^{\bm m_{\bm j}}\bar {q}_{\bm j}^{\bm m_{\bm j}'}.
\end{align*}

Since $r>2$ and $0<\sigma<r/2$, one has
\begin{align}\label{022202}
1<r-\sigma<r.
\end{align}
In view of \eqref{022202} and
\begin{align*}
\Delta (\bm l)\leq \Delta(\bm n)+\Delta(\bm m),
\end{align*} one has
\begin{align}
&\nonumber\sum_{\bm l\in\mathbb{N}^{\mathbb{Z}^d}\times\mathbb{N}^{\mathbb{Z}^d}}
\left| \sum_{\bm k\in\mathbb{Z}^d} \sum_{\bm n,\bm m\in\mathbb{N}^{\mathbb{Z}^d}\times\mathbb{N}^{\mathbb{Z}^d}}^*W(\bm n)U(\bm m)\left(\bm n_{\bm k}\bm m'_{\bm k}-\bm n_{\bm k}'\bm m_{\bm k}\right) \right|\left(r-\sigma\right)^{\Delta(\bm l)} \\
\leq&\nonumber\sum_{\bm n,\bm m\in\mathbb{N}^{\mathbb{Z}^d}\times\mathbb{N}^{\mathbb{Z}^d}}
 \left|W(\bm n)\right|\left|U(\bm m)\right|\sum_{\bm k\in\mathbb{Z}^d}\left(\bm n_{\bm k}\bm m'_{\bm k}+\bm n_{\bm k}'\bm m_{\bm k}\right) \left(r-\sigma\right)^{\Delta(\bm n)+\Delta(\bm m)} \\
\leq&\label{022203}\left(\sum_{\bm n\in\mathbb{N}^{\mathbb{Z}^d}\times\mathbb{N}^{\mathbb{Z}^d}}|W(\bm n)|\cdot |\bm n|_1\cdot r^{\Delta(\bm n)}\right)
\left(\sum_{\bm m\in\mathbb{N}^{\mathbb{Z}^d}\times\mathbb{N}^{\mathbb{Z}^d}}|U(\bm m)|\cdot |\bm m|_1\cdot r^{\Delta(\bm m)}\right).
\end{align}
In view of (\ref{011401}), (\ref{022203}) and using
$|\bm l|_1=|\bm n|_1+|\bm m|_1-2,$ we have
\begin{align*}
&\left\|\left\{W,U\right\}\right\|_{j_0,N,r-\sigma}\\
\leq&\left(|\bm n|_1+|\bm m|_1-2\right)(r-\sigma)^{|\bm n|_1+|\bm m|_1-3}\\
\ \ &\times\left(\sum_{\bm n\in\mathbb{N}^{\mathbb{Z}^d}\times\mathbb{N}^{\mathbb{Z}^d}}|W(\bm n)|\cdot |\bm n|_1\cdot r^{\Delta(\bm n)}\right)
\left(\sum_{\bm m\in\mathbb{N}^{\mathbb{Z}^d}\times\mathbb{N}^{\mathbb{Z}^d}}|U(\bm m)|\cdot |\bm m|_1\cdot r^{\Delta(\bm m)}\right)\\
\leq& \frac{1}\sigma\left\|W\right\|_{j_0,N,r}\cdot\left\|U\right\|_{j_0,N,r},
\end{align*}
which finishes the proof of (\ref{122002}), and where the last inequality is based on the following inequality:
\begin{align*}
	(|\bm n|_1+|\bm m|_1-2)(r-\sigma)^{|\bm n|_1+|\bm m|_1-3}\leq \frac{1}{\sigma}r^{|\bm n|_1+|\bm m|_1-2}.
\end{align*}
\end{proof}
Furthermore, one has 
\begin{prop}\label{E1}
Let $W$ and $U$ be given by Proposition \ref{090603}.  Assume further that
\begin{align}\label{042801}
\left(\frac{e}{\sigma}\right)\left\|W\right\|_{j_0,N,r}\leq \frac12.
\end{align}
Then
\begin{align*}
\left\|U\circ X_W^1\right\|_{j_0,N,r-\sigma}
\leq2\left\|U\right\|_{j_0,N,r},
\end{align*}
where $X_W^1$ is the time-$1$ map generated by the flow of $W$.\end{prop}
In general, we have
\begin{align}\label{9}
\left\|U\circ X_{W}^1-U\right\|_{j_0,N,r-\sigma}\leq\frac{e}{\sigma}\cdot
\left\|W\right\|_{j_0,N,r}\cdot
\left\|U\right\|_{j_0,N,r},
\end{align}
and
\begin{align}\label{10}
\left\|U\circ X_W^1-U-\{U,W\}\right\|_{j_0,N,r-\sigma}\leq\left(\frac{e}{\sigma}\right)^2\left\|W\right\|_{j_0,N,r}^2\cdot
\left\|U\right\|_{j_0,N,r}.
\end{align}

\section{Analysis  of the Symplectic Transformations}

We now construct the symplectic transformation $\Gamma$ by a finite step induction.

At the first step, i.e., $s=1$ (in view of (\ref{010703})),
\begin{align}\label{011410}
H_1=H=\frac{1}2\left(\sum_{\bm j\in\mathbb{Z}^d}V_{\bm j}(\bm\theta,\bm\alpha)|q_{\bm j}|^2+\epsilon_1\sum_{\bm i,\bm j\in\mathbb{Z}^d\atop |\bm i-\bm j|_1=1}q_{\bm i}\bar q_{\bm j}+\frac12\epsilon_2\sum_{\bm j\in\mathbb{Z}^d}|q_{\bm j}|^4\right).
\end{align}
Then (\ref{011410}) can be rewritten as
\begin{align*}
\nonumber H_1&=D+Z_1+R_1, 
\end{align*}
where
\begin{equation*}
D=\frac12\sum_{\bm j\in\mathbb{Z}^d}V_{\bm j}(\bm\theta,\bm\alpha)|q_{\bm j}|^2,
\end{equation*}
\begin{equation*}
Z_1=\frac{\epsilon_2}4\sum_{\bm j\in\mathbb{Z}^d}|q_{\bm j}|^4,
\end{equation*}
and 
\begin{equation*}
R_1=\frac{\epsilon_1}2\sum_{\bm i,\bm j\in\mathbb{Z}^d\atop |\bm i-\bm j|=1}q_{\bm i}\bar q_{\bm j}.
\end{equation*}
Let $M$ be given as in Theorem \ref{main} and fix any $r>2$. From (\ref{122402}), we have that
\begin{align}\label{122502}
\left\|H_1-D\right\|_{j_0,M^2,r}\leq \epsilon^{0.99},
\end{align}
using $\epsilon=\epsilon_1+\epsilon_2$ small enough.

\subsection{One step of Birkhoff normal form}

\begin{lem}\label{122501}
Let $V=(V_{\bm j}(\bm\theta,\bm\alpha))_{\bm j\in\mathbb{Z}^d}$ satisfy the $(\gamma,L,M,j_0)$-nonresonant conditions \eqref{122106}. 
Then there exists a change of variables $\Gamma_1:=X_{F_1}^1$ such that
\begin{align*}
\nonumber H_{2}&=H_1\circ X_{F_1}^1=D+Z_2+R_2\\
&=\frac12\sum_{\bm j\in\mathbb{Z}^d}V_{\bm j}(\bm\theta,\bm\alpha)|q_{\bm j}|^2
+\sum_{\bm n\in\mathbb{N}^{\mathbb{Z}^d}\times\mathbb{N}^{\mathbb{Z}^d}\atop \bm n\in\mathcal{N}}Z_{2}(\bm n)\prod_{{\rm supp}\ \bm n}\left|q_{\bm j}^{\bm n_{\bm j}}\right|^2
 +\sum_{\bm n\in\mathbb{N}^{\mathbb{Z}^d}\times\mathbb{N}^{\mathbb{Z}^d}}R_{2}(\bm n)\prod_{{\rm supp}\ \bm n}q_{\bm j}^{\bm n_{\bm j}}\bar q_{\bm j}^{\bm n_{\bm j}'}.
\end{align*}
Moreover, one has
\begin{align}
\label{122507}\left\|F_1\right\|_{j_0,M^2,r}&\leq\epsilon^{0.95},\\
\label{122601}\left\|Z_2\right\|_{j_0,M^2,r-\sigma}&\leq \epsilon^{0.9}\left(\sum_{i=0}^{1}2^{-i}\right),\\
\label{122603}\left\|R_2\right\|_{j_0,M^2,r-\sigma}&\leq \epsilon^{0.9}\left(\sum_{i=0}^{1}2^{-i}\right),
\end{align}
and
\begin{align}\label{011403}
\left\|\mathcal{R}_2\right\|_{j_0,M^2,r-\sigma}\leq \epsilon^{1.9},
\end{align}
where
\begin{align}\label{122604}
\mathcal{R}_2=\sum_{\bm n\in\mathbb{N}^{\mathbb{Z}^d}\times\mathbb{N}^{\mathbb{Z}^d}}R_{2}(\bm n)\prod_{{\rm supp}\ \bm n\cap A(j_0,M^2-40)\neq\emptyset}q_{\bm j}^{\bm n_{\bm j}}\bar q_{\bm j}^{\bm n_{\bm j}'}.
\end{align}
Furthermore, for any $A\geq 3$ the following estimate holds
\begin{align}\label{122605}
\left\|\sum_{\Delta(\bm n)+|\bm n|=A}\left(|Z_2(\bm n)|+|R_2(\bm n)|\right)\prod_{{\rm supp}\ \bm n}q_{\bm j}^{\bm n_{\bm j}}\bar q_{\bm j}^{\bm n_{\bm j}'}\right
\|_{j_0,M^2,r-\sigma}
\leq\epsilon^{1+0.9(A-3)}.
\end{align}
\end{lem}

\begin{proof}
By the Birkhoff normal form theory,  $F_1$ satisfies the homological equation
\begin{align}\label{hom1}
L_{V} F_1=\mathcal{R}_1,
\end{align}
where the \textit{Lie derivative} operator is defined by
\begin{align*}
L_{V}:\;W\mapsto L_{V}W:={\rm{i}}\sum_{\bm n\in\mathbb{N}^{\mathbb{Z}^d}\times\mathbb{N}^{\mathbb{Z}^d}}\left(\sum_{\bm j\in\mathbb{Z}^d}(\bm n_{\bm j}-\bm n_{\bm j}')V_{\bm j},\right)W(\bm n)\prod_{{\rm supp}\ \bm n}q_{\bm j}^{\bm n_{\bm j}}\bar q_{\bm j}^{\bm n_{\bm j}'},
\end{align*}
and
\begin{align*}
\mathcal{R}_1=\sum_{\bm n\in\mathbb{N}^{\mathbb{Z}^d}\times\mathbb{N}^{\mathbb{Z}^d}}R_1(\bm n)\prod_{{\rm supp}\ \bm n\cap A(j_0,M^2-20)\neq \emptyset}q_{\bm j}^{\bm n_{\bm j}}\bar q_{\bm j}^{\bm n_{\bm j}'}.
\end{align*}
Unless $\bm n\in\mathcal{N}$ (see (\ref{030301})), one has
\begin{align*}
F_1(\bm n)=\frac{R_1(\bm n)}{\sum_{\bm j\in\mathbb{Z}^d}(\bm n_{\bm j}-\bm n_{\bm j}')V_{\bm j}}.
\end{align*}
Note that frequency $V$ satisfies the nonresonant condition (\ref{122106}).  Using $\epsilon$ small enough we have
\begin{align}\label{122503} |F_1(\bm n)|\leq {|R_1(\bm n)|}\epsilon^{-0.01},\end{align}
and
then
\begin{eqnarray*}
\left\|F_1\right\|_{j_0,M^2,r}\leq\epsilon^{0.95},
\end{eqnarray*}which completes the proof of (\ref{122507}).

Using Taylor's formula yields
\begin{align*}
\nonumber H_{2}&:=H_1\circ X_{F_1}^1=D+Z_1\\
\nonumber &\ +\{D,F_1\}+R_1+\left(X_{F_1}^1-\textbf{id}-\{\cdot,F_1\}\right)D
+\left(X_{F_1}^1-\textbf{id}\right)(Z_1+R_1)\\
\nonumber&=D+Z_2+R_2\\
&:=\frac12\sum_{\bm j\in\mathbb{Z}^d}V_{\bm j}|q_{\bm j}|^2
+\sum_{\bm n\in\mathbb{N}^{\mathbb{Z}^d}\times\mathbb{N}^{\mathbb{Z}^d}\atop \bm n\in\mathcal{N}}Z_{2}(\bm n)\prod_{{\rm supp}\ n}\left|q_{\bm j}^{\bm n_{\bm j}}\right|^2
+\sum_{\bm n\in\mathbb{N}^{\mathbb{Z}^d}\times\mathbb{N}^{\mathbb{Z}^d}}R_{2}(\bm n)\prod_{{\rm supp}\ \bm n}q_{\bm j}^{\bm n_{\bm j}}\bar q_{\bm j}^{\bm n_{\bm j}'},
\end{align*}
where by \eqref{hom1},
\begin{align}
\nonumber R_2 &=(R_1-\mathcal{R}_1)+\left(X_{F_1}^1-\textbf{id}-\{\cdot,F_1\}\right)D+\left(X_{F_1}^1-\textbf{id}\right)(Z_1+R_1)\\
\nonumber &=\sum_{\bm n\in\mathbb{N}^{\mathbb{Z}^d}\times\mathbb{N}^{\mathbb{Z}^d}}R_{2}(\bm n)\prod_{{\rm supp}\ \bm n}q_{\bm j}^{\bm n_{\bm j}}\bar q_{\bm j}^{\bm n_{\bm j}'},
\end{align}
and
\begin{align*}
&\left(X_{F_1}^1-\textbf{id}-\{\cdot,F_1\}\right)D:=D\circ X_{F_1}^1-D-\{D,F_1\},\\
&\left(X_{F_1}^1-\textbf{id}\right)(Z_1+R_1):=(Z_1+R_1)\circ X_{F_1}^1-(Z_1+R_1).
\end{align*}
In this step, we have $Z_2=Z_1$. Hence, the estimate (\ref{122601}) holds true.

Write
\begin{align*}
R_2=\mathcal{R}_2+(R_2-\mathcal{R}_2),
\end{align*}
where $\mathcal{R}_2$ is defined by (\ref{122604}).
By (\ref{9}), (\ref{10}) and  \eqref{hom1}, for any $0<\sigma<r/2$ one has
\begin{align*}
\left\|\mathcal{R}_2\right\|_{j_0,M^2,r-\sigma}&\leq\nonumber \left(\frac{e}{\sigma}\right)\cdot\left\|F_1\right\|_{j_0,M^2,r}\cdot\left\|H_1-D\right\|_{j_0,M^2,r}\leq\epsilon^{1.9},
\end{align*}
where the last inequality follows from (\ref{122502}) and (\ref{122507}). This finishes the proof of (\ref{011403}).
Similarly, we have
\begin{align*}
\left\|R_2-\mathcal{R}_2\right\|_{j_0,M^2,r-\sigma}\leq \epsilon^{0.9}\left(\sum_{i=0}^{1}2^{-i}\right).
\end{align*}

Finally, the estimate  (\ref{122605}) follows from (\ref{122502}) and (\ref{122507}) by induction about $A$. Precisely,
the term in $R_2$ comes from $\frac1{j!}Z_1^{(j)}$ and $\frac1{ {j}!}R_1^{({j})}$ for some $j\in \mathbb{N}$, where
$Z_1^{(j)}=\left\{Z_1^{(j-1)},H\right\}$, $Z_1^{(0)}=Z_1$, $R_1^{(j)}=\left\{R_1^{(j-1)},H\right\}$ and  $R_1^{(0)}=R_1$. Following the proof of (\ref{011403}) and noting that
$\Delta(\bm l)\leq \Delta(\bm n)+\Delta(\bm m)$ and $|\bm l|_1\leq |\bm n|_1+|\bm m|_1-2,$ we conclude the proof of (\ref{122605}).
\end{proof}

\subsection{Iterative Lemma}
We introduce some constants during the iterative steps.  For $s\in\mathbb{N}$ and $1\leq s\leq M$, let
\begin{align}\label{011601}
	N_s=M^2-20(s-1),
\end{align}
and then using $M\gg 1$ one has 
\begin{equation*}
	N_{s}\geq M^2-20M\geq \frac{M^2}2,
\end{equation*}
which implies
\begin{equation*}
	\left[j_0-\frac{M^2}2,j_0+\frac{M^2}2\right]\subset A(j_0,N_s)\subset\left[j_0-{M^2},j_0+{M^2}\right].
\end{equation*}
Let 
\begin{equation*}
	\sigma=\frac{r}{2M},
\end{equation*}
then
\begin{equation*}
	r-s\sigma\geq r-M\cdot \frac{r}{2M}\geq r/2.
\end{equation*}
\begin{lem}\label{62}
Consider the Hamiltonian $H_s(q,\bar q)$ of the form
\begin{align*}
\nonumber H_s&=D+Z_s+R_s\\
&=\frac12\sum_{\bm j\in\mathbb{Z}^d}V_{\bm j}|q_{\bm j}|^2+\sum_{\bm n\in\mathbb{N}^{\mathbb{Z}^d}\times\mathbb{N}^{\mathbb{Z}^d}\atop \bm n\in\mathcal{N}}Z_s(\bm n)\prod_{{\rm supp}\ \bm n}\left|q_{\bm j}^{\bm n_{\bm j}}\right|^2
+\sum_{\bm n\in\mathbb{N}^{\mathbb{Z}^d}\times\mathbb{N}^{\mathbb{Z}^d}}R_s(\bm n)\prod_{{\rm supp}\ \bm n}q_{\bm j}^{\bm n_{\bm j}}\bar q_{\bm j}^{\bm n_{\bm j}'}.
\end{align*}
Let $V$ satisfy the $(\gamma,L,M,j_0)$-nonresonant conditions \eqref{122106}. Assume 
\begin{align}
&\label{122601..}\left\|Z_s\right\|_{j_0,M^2,r-(s-1)\sigma}\leq \epsilon^{0.9}\left(\sum_{i=0}^{s-1}2^{-i}\right),\\
&\label{122603..}\left\|R_s\right\|_{j_0,M^2,r-(s-1)\sigma}\leq \epsilon^{0.9}\left(\sum_{i=0}^{s-1}2^{-i}\right),\\
&\label{122602..}\left\|\mathcal{R}_s\right\|_{j_0,M^2,r-(s-1)\sigma}\leq \epsilon^{1+0.9(s-1)},
\end{align}
where
\begin{align}\label{122604.}
\mathcal{R}_s=\sum_{\bm n\in\mathbb{N}^{\mathbb{Z}^d}\times\mathbb{N}^{\mathbb{Z}^d}}R_{s}(\bm n)\prod_{{\rm supp}\ \bm n\cap A(j_0,N_{s+1})\neq\emptyset}q_{\bm j}^{\bm n_{\bm j}}\bar q_{\bm j}^{\bm n_{\bm j}'}.
\end{align}
Furthermore, assume for any $A\geq 3$ the following holds
\begin{align}\label{122605.}\left\|\sum_{\Delta(\bm n)+|\bm n|=A}\left(|Z_s(\bm n)|+|R_s(\bm n)|\right)\prod_{{\rm supp}\ \bm n}q_{\bm j}^{\bm n_{\bm j}}\bar q_{\bm j}^{\bm n_{\bm j}'}\right\|_{j_0,M^2,r-(s-1)\sigma}
\leq \epsilon^{1+0.9(A-3)}.
\end{align}
Then there exists a change of variables $\Phi_s:=X_{F_s}^1$
\begin{align*}
\nonumber H_{s+1}&=H_s\circ X_{F_s}^1\\
&=\frac12\sum_{\bm j\in\mathbb{Z}^d}V_{\bm j}|q_{\bm j}|^2\\
&\quad
+\sum_{\bm n\in\mathbb{N}^{\mathbb{Z}^d}\times\mathbb{N}^{\mathbb{Z}^d}\atop \bm n\in\mathcal{N}}Z_{s+1}(\bm n)\prod_{{\rm supp}\ \bm n}\left|q_{\bm j}^{\bm n_{\bm j}}\right|^2 +\sum_{\bm n\in\mathbb{N}^{\mathbb{Z}^d}\times\mathbb{N}^{\mathbb{Z}^d}}R_{s+1}(\bm n)\prod_{{\rm supp}\ \bm n}q_{\bm j}^{\bm n_{\bm j}}\bar q_{\bm j}^{\bm n_{\bm{j}}'}.
\end{align*}
Moreover, one has
\begin{align}
&\label{122610} \left\|F_s\right\|_{j_0,M^2,r-(s-1)\sigma}\leq \epsilon^{0.9s},\\
&\label{122616}\left\|Z_{s+1}\right\|_{j_0,M^2,r-s\sigma}\leq \epsilon^{0.9}\left(\sum_{i=0}^{s}2^{-i}\right),\\
&\label{122615}\left\|R_{s+1}\right\|_{j_0,M^2,r-s\sigma}\leq \epsilon^{0.9}\left(\sum_{i=0}^{s}2^{-i}\right),\\
&\label{122611}\left\|\mathcal{R}_{s+1}\right\|_{j_0,M^2,r-s\sigma}\leq\epsilon^{1+0.9s},
\end{align}
where
\begin{align*}
\mathcal{R}_{s+1}=\sum_{\bm n\in\mathbb{N}^{\mathbb{Z}^d}\times\mathbb{N}^{\mathbb{Z}^d}}R_{s+1}(\bm n)\prod_{{\rm supp}\ \bm n\cap A(j_0,N_{s+2})\neq \emptyset}q_{\bm j}^{\bm n_{\bm j}}\bar q_{\bm j}^{\bm n_{\bm j}'}.
\end{align*}
Moreover, we have
\begin{align}
\left\|\sum_{\Delta(\bm n)+|\bm n|=A}\left(|Z_{s+1}(\bm n)|+|R_{s+1}(\bm n)|\right)\prod_{{\rm supp}\ \bm n}q_{\bm j}^{\bm n_{\bm j}}\bar q_{\bm j}^{\bm n_{\bm j}'}\right\|_{j_0,M^2,r-s\sigma}
\label{122617}\leq\epsilon^{1+0.9(A-3)}.
\end{align}

\end{lem}
\begin{proof}
As before, $F_s$ satisfies the homological equation
\begin{align*}
L_{V} F_s=\widetilde{\mathcal R}_s,
\end{align*}
where
\begin{align}\label{022501}
\widetilde{\mathcal{R}}_s(q,\bar q):=\sum_{\bm n\in\mathbb{N}^{\mathbb{Z}^d}\times\mathbb{N}^{\mathbb{Z}^d}}R_{s}(\bm n)\prod_{{\rm supp}\ \bm n\cap A(j_0,N_{s+1})\neq \emptyset\atop
\Delta (\bm n)+|\bm n|\leq s+2}q_{\bm j}^{\bm n_{\bm j}}\bar q_{\bm j}^{\bm n_{\bm j}'}.
\end{align}
By direct computations, one has
\begin{align*}
F_s(\bm n)=\frac{R_s(\bm n)}{\sum_{\bm j\in\mathbb{Z}^d}(\bm n_{\bm j}-\bm n_{\bm j}')V_{\bm j}},
\end{align*}
unless $\bm n\in\mathcal{N}$.
In view of (\ref{122602..}) and following the proof of (\ref{122507}), 
we have
\begin{align}
\nonumber\left\|F_s\right\|_{j_0,M^2,r-(s-1)\sigma}\leq \epsilon^{0.9s},
\end{align}
which finishes the proof of (\ref{122610}).


Using Taylor's formula again shows
\begin{align*}
H_{s+1}&:=H_s\circ X_{F_s}^1\\
&=D+\{D,F_s\}+Z_s+R_s+\left(X_{F_s}^1-\textbf{id}-\{\cdot,F_s\}\right)D
+\left(X_{F_s}^1-\textbf{id}\right)(Z_s+R_s)\\
&=D+Z_{s+1}+R_{s+1}\\
&=\frac12\sum_{\bm j\in\mathbb{Z}^d}V_{\bm j}|q_{\bm j}|^2+\sum_{\bm n\in\mathbb{N}^{\mathbb{Z}^d}\times\mathbb{N}^{\mathbb{Z}^d}\atop
\bm n\in\mathcal{N}}Z_{s+1}(\bm n)\prod_{{\rm supp}\ \bm n}\left|q_{\bm j}^{\bm n_{\bm j}}\right|^2 +\sum_{\bm n\in\mathbb{N}^{\mathbb{Z}^d}\times\mathbb{N}^{\mathbb{Z}^d}}R_{s+1}(\bm n)\prod_{{\rm supp}\ \bm n}q_{\bm j}^{\bm n_{\bm j}}\bar q_{\bm j}^{\bm n_{\bm j}'},
\end{align*}
where
\begin{align}
\nonumber &Z_{s+1}=Z_s+\sum_{\bm n\in\mathbb{N}^{\mathbb{Z}^d}\times\mathbb{N}^{\mathbb{Z}^d}\atop \bm n\in\mathcal{N}}R_{s}(\bm n)\prod_{{\rm supp}\ \bm n\cap A(j_0,N_{s+1})\neq \emptyset\atop
	\Delta (\bm n)+|\bm n|\leq s+2}q_{\bm j}^{\bm n_{\bm j}}\bar q_{\bm j}^{\bm n_{\bm j}'}.
\end{align}
Following the proof of (\ref{122601..})-(\ref{122602..}), one completes the proof of (\ref{122616})-(\ref{122611}).



Finally, the estimate (\ref{122617}) follows from the proof of (\ref{122605}).


\end{proof}
\subsection{Birkhoff Normal Form Theorem}
\begin{thm}[\textbf{Birkhoff Normal Form}]\label{011501} Consider the Hamiltonian \eqref{011410} and assume that the potential $V$ satisfies the $(\gamma,L,M,j_0)$-nonresonant condition \eqref{122106}. Given any $r>2$, there exists an $\varepsilon(\gamma,L,M,j_0)>0$ such that, for any $0<\epsilon<\varepsilon(\gamma,L,M,j_0)$,
there exists a symplectic transformation $\Gamma=\Gamma_1\circ\cdots\circ\Gamma_M$
such that
\begin{align*}
\widetilde H&=H_1\circ \Gamma={D}+\widetilde{Z}+\widetilde{R}\\
&=\frac12\sum_{\bm j\in\mathbb{Z}^d}V_{\bm j}|q_{\bm j}|^2+\sum_{\bm n\in\mathbb{N}^{\mathbb{Z}^d}\times\mathbb{N}^{\mathbb{Z}^d}\atop \bm n\in\mathcal{N}}\widetilde Z(\bm n)\prod_{{\rm supp}\ \bm n}q_{\bm j}^{\bm n_{\bm j}}\bar q_{\bm j}^{\bm n_{\bm j}'}\\
&\ +\sum_{\bm n\in\mathbb{N}^{\mathbb{Z}^d}\times\mathbb{N}^{\mathbb{Z}^d}}\widetilde R(\bm n)\prod_{{\rm supp}\ \bm n}q_{\bm j}^{\bm n_{\bm j}}\bar q_{\bm j}^{\bm n_{\bm j}'},
\end{align*}
where
\begin{align}
&\label{030401} \left\|\widetilde {Z}\right\|_{j_0,M^2,r/2}\leq 2\epsilon^{0.9},\\
&\label{030402} \left\|\widetilde R\right\|_{j_0,M^2,r/2}\leq 2\epsilon^{0.9},
\end{align}
and
\begin{align}
\label{030403}&\left\|\widetilde{\mathcal{R}}\right\|_{j_0,M^2,r/2}\leq\epsilon^{0.9M},
\end{align}
with
\begin{align}
\label{042301}&\widetilde{\mathcal{R}}=\sum_{\bm n\in\mathbb{N}^{\mathbb{Z}^d}\times\mathbb{N}^{\mathbb{Z}^d}}\widetilde{R}(\bm n)\prod_{{\rm supp}\ \bm n\cap A(j_0,M^2/2)\neq \emptyset}q_{\bm j}^{\bm n_{\bm j}}\bar q_{\bm j}^{\bm n_{\bm j}'}.
\end{align}
Furthermore, for any $A\geq 3$ the following estimate holds
\begin{align}
\left\|\sum_{\Delta(\bm n)+|\bm n|=A}\left(|\widetilde{Z}(\bm n)|+|\widetilde{R}(\bm n)|\right)\prod_{{\rm supp}\ \bm n}q_{\bm j}^{\bm n_{\bm j}}\bar q_{\bm j}^{\bm n_{\bm j}'}\right\|_{j_0,M^2,r/2} 
\label{042302}\leq\epsilon^{1+0.9(A-3)}.
\end{align}
\end{thm}

\begin{proof}
First of all, note that the Hamiltonian (\ref{011410}) satisfies all assumptions (\ref{122601..})--(\ref{122605.}) for $s=1$, which follows from (\ref{122502}).

Finally, using the Iterative Lemma, one can find a symplectic transformation $\Gamma=\Gamma_1\circ\cdots\circ\Gamma_M$ such that
\begin{align*}
\widetilde H:=H_{M+1}=H_1\circ \Gamma, 
 \end{align*}which satisfies (\ref{030401})-(\ref{042302}).
\end{proof}

\subsection{Proof of the main theorem}

Now we are in a position to complete the proof of Theorem \ref{main}.

\begin{proof}
	In view of Theorem \ref{011501}, one obtains
	the  $\widetilde H(\tilde q,\bar{\tilde q})$ in the new coordinates. The new Hamiltonian equation is given by
	\begin{align}\label{011503}
		{\rm {i}} \dot{\tilde q}=2\frac{\partial\widetilde{H}}{\partial \bar{\tilde q}}.
	\end{align}
	We get by using (\ref{011503}) that
	\begin{align*}
		\frac{d}{dt}\sum_{|\bm j|>j_0}\left|\tilde q_{\bm j}(t)\right|^2=&\left\{\sum_{|\bm j|>j_0}\left|\tilde q_{\bm j}(t)\right|^2,\widetilde D+\widetilde Z+\widetilde R\right\}\\
		=&\left\{\sum_{|\bm j|>j_0}\left|\tilde q_{\bm j}(t)\right|^2,\widetilde R\right\}\\=&4\mbox{Im}  \sum_{|\bm j|>j_0}\bar{\tilde q}_{\bm j}(t)\frac{\partial\widetilde{R}}{\partial \bar{\tilde q}}\\
		=&\sum_{\bm n\in\mathbb{N}^{\mathbb{Z}^d}\times\mathbb{N}^{\mathbb{Z}^d}}
		\widetilde{R}(\bm n)\sum_{|\bm j|>j_0}(\bm n_{\bm j}-\bm n_{\bm j}')\prod_{{\rm supp}\ \bm n}{\tilde q_{\bm j}}^{\bm n_{\bm j}}\bar {\tilde q}_{\bm j}^{\bm n_{\bm j}'}.
	\end{align*}
	In view of (\ref{042301}), we decompose $\widetilde R$ into  three parts:
	\begin{align*}
		\widetilde R=\widetilde {{R}}^{(1)}+ \widetilde {{R}}^{(2)}+\widetilde {{R}}^{(3)},
	\end{align*}
	where
	\begin{align}
		\nonumber&\widetilde {{R}}^{(1)}=\widetilde{\mathcal{R}},\\
		\nonumber&\widetilde {{R}}^{(2)}=\sum_{\bm n\in\mathbb{N}^{\mathbb{Z}^d}\times\mathbb{N}^{\mathbb{Z}^d}}
		\widetilde{R}(\bm n)\sum_{|\bm j|>j_0}(\bm n_{\bm j}-\bm n_{\bm j}')\prod_{{\rm supp}\ \bm n\cap A(j_0,M^2/2)=\emptyset\atop \Delta(\bm n)\geq M+4}{\tilde q_{\bm j}}^{\bm n_{\bm j}}\bar {\tilde q}_{\bm j}^{\bm n_{\bm j}'},\\
		\label{042304}&\widetilde {{R}}^{(3)}=\sum_{\bm n\in\mathbb{N}^{\mathbb{Z}^d}\times\mathbb{N}^{\mathbb{Z}^d}}
		\widetilde{R}(\bm n)\sum_{|\bm j|>j_0}(\bm n_{\bm j}-\bm n_{\bm j}')\prod_{{\rm supp}\ n\cap A(j_0,M^2/2)=\emptyset\atop \Delta(\bm n)\leq M+3}{\tilde q_{\bm j}}^{\bm n_{\bm j}}\bar {\tilde q}_{\bm j}^{\bm n_{\bm j}'}.
	\end{align}
	Using (\ref{030403}) and (\ref{042302}) implies
	\begin{align}\label{042303}
		\left\|\widetilde{{R}}^{(1)}+\widetilde{{R}}^{(2)}\right\|_{j_0,M^2,r/2}\leq\epsilon^{M+1}.
	\end{align}
	
	Now consider the monomials in $\widetilde{{R}}^{(3)}$. Recalling that
	\begin{align*}
		\Delta(\bm n)\leq M+3,
	\end{align*}
	if ${\rm supp}\ \bm n\cap A(j_0,M^2/2)= \emptyset$, then
	for any $\bm j\in {\rm supp}\ \bm n$ satisfying 
	\begin{align*}
	|\bm j|>j_0.
	\end{align*}
	Hence the terms in (\ref{042304}) satisfy
	\begin{align}\label{042305}
		\sum_{|\bm j|>j_0}(\bm n_{\bm j}-\bm n_{\bm j}')=0.
	\end{align}
	Using (\ref{042303}) and (\ref{042305}), one has
	\begin{align*}
		\frac{d}{dt}\sum_{|j|>j_0}\left|\tilde q_j(t)\right|^2\leq \epsilon^{M+1}.
	\end{align*}

	Integrating in $t$, we obtain
	\begin{equation}\label{051801}
		\sum_{|\bm j|>j_0}\left|\tilde q_{\bm j}(t)\right|^2\leq \sum_{|\bm j|>j_0}\left|\tilde q_{\bm j}(0)\right|^2+\epsilon^{M+1}t.
	\end{equation}
	
	Note that the symplectic transformation only acts on the $M^2$-neighborhood of $\left\{\bm j\in\mathbb{Z}^d:\ |\bm j|=j_0\right\}$.  We obtain
	\begin{align*}
		\sum_{|\bm j|>j_0+M^2}|q_{\bm j}(t)|^2\leq\sum_{|\bm j|>j_0}|\tilde q_{\bm j}(t)|^2,
	\end{align*}
	which together with (\ref{051801}) gives
	\begin{align*}
		\sum_{|\bm j|>j_0+M^2}|q_{\bm j}(t)|^2\leq \sum_{|\bm j|>j_0}\left|\tilde q_{\bm j}(0)\right|^2+\epsilon^{M+1}t.
	\end{align*}
	
	On the other hand, the Hamiltonian preserves the $\ell^2$-norm. So we have
	\begin{align*}
		\sum_{|\bm j|>j_0}|\tilde q_{\bm j}(0)|^2=\sum_{\bm j\in\mathbb{Z}^d}|q_{\bm j}(0)|^2-\sum_{|\bm j|\leq j_0}|\tilde q_{\bm j}(0)|^2<\sum_{|\bm j|>j_0-M^2}|q_{\bm j}(0)|^2.
	\end{align*}
Hence one has
\begin{align*}
	\sum_{|\bm j|>j_0+M^2}|q_{\bm j}(0)|^2\leq 2\delta,
\end{align*}
for
	\begin{align*}\label{051803}
		|t|\leq \delta\cdot \epsilon^{-M}.
	\end{align*}

\end{proof}
\section{Estimates on the measure}
Let
\begin{equation*}
	\omega=\left(\omega_{\bm j}\right)_{\bm j\in\mathbb{Z}^d}
\end{equation*}
with 
\begin{equation*}
	\omega_{\bm j}=	V_{\bm j}(\bm\theta,\bm\alpha),
\end{equation*}
which is given by (\ref{080401}).
Define the resonant set $\mathfrak{R}( \bm k)$ with $0\neq \bm k\in\mathbb{Z}^{\mathbb{Z}^d}$ by
\begin{align*}
	\mathfrak{R}\left(\bm k\right)=\left\{\left(\bm \theta,\bm \alpha\right):\ \left|\sum_{\bm j\in\mathbb{Z}^d}k_{\bm j}\omega_{\bm j}\right|
	<\left(\frac{\gamma}{\left(2LM^2j_0\right)^{2(d+1)}}\right)^{10L^4M^4}\right\},
\end{align*}
and 
\begin{align*}
	\mathfrak{R}(\gamma,L,M,j_0)=\bigcup_{\bm k}^{*}\mathfrak{R}(\bm k),
\end{align*}
where the union $$\bigcup_{\bm k}^{*}$$
is taken for $\bm k$ satisfying 
\begin{equation}\label{080402}
	{\rm supp}\ \bm k\cap A(j_0,M^2)\neq \emptyset
	\end{equation}and
\begin{equation}\label{080403}\Delta (\bm k)+|\bm k|\leq M+2.
	\end{equation}

Denote by ${\rm meas(\cdot)}$ the Lebesgue measure. Then one has 
\begin{lem}[]\label{nlsthm}
	\begin{align}\label{090402}
		{\rm meas}(	\mathfrak{R}(\gamma,L,M,j_0))\leq \gamma.
	\end{align}
\end{lem}
\smallskip
The proof of Lemma~\ref{nlsthm} is similar to that in \cite{SW23-2}. For the reader's convenience,
we reproduce the main steps below. For completeness, we also include the main lemmas used in 
\cite{SW23-2} in the Appendix. 
\smallskip
\begin{proof}
Note that the number of $\bm k$ satisfying (\ref{080402}) and (\ref{080403}) is less than $$\left(j_0M^2\right)^{d}\cdot(2M+1)^{d(M+2)}\leq \left(j_0M^2\right)^{2dM}.$$ 
Therefore it suffices to estimate $	{\rm meas}(	\mathfrak{R}(\bm k)).$

Recalling (\ref{080401}), one has 
\begin{align*}
\sum_{\bm j\in\mathbb{Z}^d}k_{\bm j}\omega_{\bm j}= (\bm k\otimes \bm v)\cdot \bm V,
\end{align*}
where 
\begin{align*}
	&\bm k\otimes \bm v=(k_{\bm j} v_{\bm{\bm \ell}})_{\bm j\in{\rm{supp}}\ {\bm k}, \   \bm\ell\in \Gamma_L},\\
	&\bm V=(\cos2\pi\bm \ell\cdot(\bm\theta+\bm j\bm\alpha))_{\bm j\in{\rm{supp}}\ {\bm k}, \   \bm\ell\in \Gamma_L}.
\end{align*}
Choose some suitable vector $\bm \xi=(\widetilde{\bm \xi}, \widehat{\bm \xi})\in (0,1]^d\times (0,1]^d$, to be given shortly, and define 
\begin{align*}
	d _{\bm \xi} =\sum_{i=1}^d\left(\widetilde{ \xi}_i\frac{\partial}{\partial \alpha_i}+\widehat{\xi}_i\frac{\partial}{\partial \theta_i}\right).
\end{align*}
Hence for $s\geq 1$, one obtains 
\begin{align*}
	&\ \ \ d _{\bm \xi}^{2s} \cos2\pi\bm  \ell\cdot (\bm\theta+\bm j\bm \alpha)\\
	&=(-1)^s(2\pi)^{2s}((\bm\ell\bm j)\cdot \widetilde{\bm \xi}+ \bm\ell\cdot \widehat{\bm \xi})^{2s} \cos2\pi\bm \ell\cdot (\bm\theta+\bm j\bm \alpha),
\end{align*}
where 
\begin{equation*}
 d _{\bm \xi}^{s+1} \cos2\pi\bm  \ell\cdot (\bm\theta+\bm j\bm \alpha)=d _{\bm \xi}\left( d _{\bm \xi}^{s} \cos2\pi\bm  \ell\cdot (\bm\theta+\bm j\bm \alpha)\right).
\end{equation*}
This  motivates us to consider the Wronskian 
\begin{align*}
	W=\left[ d _{\bm \xi}^{2s}  {V}_{(\bm j,\bm \ell)} \right]_{(\bm j, \bm \ell)}^{1\leq s\leq R}\ {\rm with}\  R=  (\#{\rm {supp} }\ \bm k)\cdot (\#\Gamma_L),
\end{align*}
which is  a $R\times R$ real matrix.  Direct computations show 
\begin{align*}
 |\det W|
	\label{wrons}=
	A_1\cdot A_2\cdot A_3,
\end{align*}
where
\begin{eqnarray*}
		A_1&=&\prod_{\bm j,\bm\ell}| \cos2\pi {\bm\ell}\cdot(\bm\theta+\bm j \bm \alpha)|,\\
		A_2&=&\prod_{\bm j, \bm\ell}(2\pi)^2((\bm\ell\bm j)\cdot \widetilde{\bm \xi}+ \bm\ell\cdot \widehat{\bm \xi})^2,
\end{eqnarray*}
and
\begin{align*}
A_3&= \prod_{ (\bm j,\bm\ell)\neq (\bm j',\bm\ell')}(2\pi)^2\left((\bm\ell\bm j)\cdot \widetilde{\bm \xi}+ \bm\ell\cdot \widehat{\bm \xi})^2-((\bm\ell'\bm j')\cdot \widetilde{\bm \xi}+ \bm\ell'\cdot \widehat{\bm \xi})^2\right|.
\end{align*}

If we choose a  Diophantine vector $\bm \xi$, then we have
\begin{equation}\label{090403}
	A_2\cdot A_3 \geq \left(\frac{\gamma}{\left(2LM^2j_0\right)^{2(d+1)}}\right)^{10L^3M^3},
\end{equation}
by using that $\Gamma_L\subset \mathbb{Z}^d$ satisfies Properties (a) and (b) (see (\ref{080401}) for the details).

We now estimate the lower bound of $A_1$.  Firstly note that there are at least $R$ terms in the product $\prod_{\bm j,\bm\ell}$ and $R\leq M(2L+1)$.  Hence it suffices to estimate 
$$| \cos2\pi {\bm\ell}\cdot(\bm\theta+\bm j \bm \alpha)|$$
for fixed $\bm j\in{\rm{supp}}\ {k}, \   \bm\ell\in \Gamma_L$. Let $\left\|x\right\|_{{\mathbb{T}}/2}={\rm dist(x,\mathbb{Z}/2)}.$
Using 
\begin{equation*}
	4\left\|x-\frac14\right\|_{{\mathbb{T}}/2}\leq \left|\cos 2\pi x\right|\leq 	2\pi\left\|x-\frac14\right\|_{{\mathbb{T}}/2},
\end{equation*}there exists a subset $\Pi_{\bm j,\bm l}$ (the union of intervals) satisfying 
	\begin{align*}\label{090401}
		{\rm meas}(\Pi_{\bm j,\bm l})\leq \frac{\gamma}{\left(2LM^2j_0\right)^{2(d+1)}},
	\end{align*}
such that for any $(\bm\theta,\bm\alpha)\in [0,1]^{2d}\setminus \Pi_{\bm j,\bm l},$ one has 
\begin{equation*}\label{090404}
	| \cos2\pi {\bm\ell}\cdot(\bm\theta+\bm j \bm \alpha)|\geq \frac{\gamma}{\left(2LM^2j_0\right)^{2(d+1)}}.
\end{equation*}

Hence for each  $(\bm\theta,\bm\alpha)\in [0,1]^{2d}\setminus	\bigcup_{\bm j,\bm l} \Pi_{\bm j,\bm l},$ one has 
\begin{equation}\label{080701}
	\left|	\det W\right|\geq \left(\frac{\gamma}{\left(2LM^2j_0\right)^{2(d+1)}}\right)^{20L^3M^3}.
\end{equation}
By (\ref{080701}) and using Lemmas \ref{BGG85} and \ref{SW23-2} in the Appendix, yields
\begin{equation*}
	\mbox{\rm meas} (\mathfrak{R}(\bm k))\leq \left(\frac{\gamma}{\left(2LM^2j_0\right)^{2(d+1)}}\right)^{10LM},
\end{equation*}
and further 
\begin{equation*}
	\mbox{\rm meas} (\mathfrak{R}(\gamma,L,M,j_0))\leq \gamma,
\end{equation*}
which finishes the proof of (\ref{090402}).
\end{proof}

\section{Appendix}
We collect below the lemmas implicated in the proof of Lemma~\ref{nlsthm}.

\begin{lem}[\cite{BGG85}]\label{BGG85}
	Let $\bm v^{(1)},\cdots, \bm v^{(r)}\in\R^r$ be $r$ linearly independent vectors with $|\bm v^{(l)}|_1\leq M$ for $1\leq l\leq r.$ Then for any  $\bm w\in\R^r$, we have
	\begin{align*}
		\max_{1\leq l\leq r}|\bm w\cdot\bm v^{(l)}|\geq r^{-3/2}M^{1-r}|\bm w|_2\cdot|\det \left[ \bm v^{(l)}\right]_{1\leq l\leq r}|.
	\end{align*}
\end{lem}
	\begin{lem}[\cite{KM98}]\label{km98}
	Let $I\subset \R$ be an interval of finite length (i.e., $0<|I|<\infty$) and $k\geq 1.$  If  $f\in C^{k}(I;\R)$ satisfies 
	\begin{align*}
		\inf_{x\in I}\left|\frac{d^k}{dx^k} f(x)\right|\geq A>0,
	\end{align*}
	then for all $\gamma>0,$
	\begin{align*}
		{\rm meas}(\{x\in I:\ |f(x)|\leq \gamma\})\leq  \zeta_k\left( \frac\gamma A\right)^{\frac1k},
	\end{align*}
	where 
	$\zeta_k=k(k+1)\left((k+1)!\right)^{\frac1k}.$
\end{lem}
\begin{lem}[\cite{SW23-2}]\label{SW23-2}
	Fix $k\in \N$ and let $I=I_{\bm a, \bm b}=\prod_{i=1}^d[a_i,b_i]\subset\R^d$. Assume that  the function $f\in C^{k+1}(I;\R)$ satisfies for some $A>0,$
	\begin{align*}\label{ass1}
		\inf_{\bm x\in I}\sup_{1\leq l\leq k}\left|{d_{\bm \beta}^l}f(\bm x)\right|\geq A,
	\end{align*}
where 
\begin{align*}
	d_{\bm \beta}:=\sum_{i=1}^d\beta_i{\partial}_i,\ \partial_i:=\frac{\partial}{\partial x_i}.
\end{align*}
with ${\bm \beta}=(\beta_1,\cdots,\beta_d)\in\R^d\setminus\{0\}$
 and $${d_{\bm \beta}^l}f(\bm x)=d_{\bm \beta}(d_{\bm \beta}^{l-1}f(\bm x))$$
 for $l\geq 1$.
	Let 
	\begin{align*}
		\|f\|_{k+1}=\sup_{\bm x\in I}\sup_{1\leq |\bm \gamma|\leq k+1}\left| {\partial^{\bm \gamma}}f(\bm x)\right|<\infty
	\end{align*}
with 
\begin{align*}
	|\bm \gamma|=\sum_{i=1}^d\gamma_i,\  \partial^{\bm \gamma}=\partial^{\gamma_1}_1\cdots\partial^{\gamma_d}_d.
\end{align*}

	Then for $0<\varepsilon<A<1,$
	\begin{align*}
		\nonumber&\ \ \ {\rm meas}(\{\bm x\in I:\ |f(\bm x)|\leq \varepsilon\})\\
		&\leq C(\bm \beta,k,d) (\|f\|_{k+1}|\bm b-\bm a|_2+1)^d {|\bm a\vee\bm b|}^{d-1}A^{-(d+\frac1k)}{\varepsilon}^{\frac{1}{k}},
	\end{align*}
	where $C=C(\bm \beta, k,d)>0$ depends only on $\bm\beta, k,d $ (but not on $f$) and 
	$$|\bm a\vee \bm b|=\sum_{i=1}^d \max(|a_i|, |b_i|).$$
\end{lem}

\begin{rem} Lemmas~\ref{BGG85} and \ref{km98} are two of the ingredients used in the proof of Lemma~\ref{SW23-2}. 
\end{rem}

\section*{Acknowledgments}
    H. Cong  was supported by NSFC (11671066).  Y. Shi  was supported by NSFC (12271380). 
   W.-M. Wang acknowledges support from the
CY Initiative of Excellence, ``Investissements d'Avenir" Grant No. ANR-16-IDEX-0008.

 \bibliographystyle{alpha}

\end{document}